\documentclass[onecolumn, 11pt, journal]{IEEEtran}
\usepackage{multirow}
\usepackage{epsfig}
\usepackage{subfigure}
\usepackage{cite}
\usepackage{amsmath}
\usepackage{amssymb}
\usepackage{footnote}
\usepackage{array}
\usepackage{subfigure, verbatim, animate}
\usepackage{comment, amsmath, amsfonts, amsthm}
\usepackage[mathscr]{euscript}
\usepackage{epsfig}
\usepackage{amsfonts,amssymb,latexsym,amsmath, amsxtra,verbatim}
\usepackage[all]{xy}
\usepackage{auto-pst-pdf}%%%
\usepackage{psfrag}
\usepackage{xcolor}
\usepackage{xspace}
\usepackage{bbm}
\input{mathlig}

\newcommand{\muspace}{\mspace{1mu}}

\DeclareRobustCommand{\scond}{\mathchoice{\muspace\vert\muspace}{\vert}{\vert}{\vert}}
\mathlig{|}{\scond}

\DeclareRobustCommand{\discint}{\mathchoice{\mspace{-1.5mu}:\mspace{-1.5mu}}{\mspace{-1.5mu}:\mspace{-1.5mu}}{:}{:}}
\mathlig{::}{\discint}

\newcommand{\Qcal}{\mathcal{Q}}

\newcommand{\Scal}{\mathcal{S}}

\newcommand{\Ucal}{\mathcal{U}}
\newcommand{\Vcal}{\mathcal{V}}

\newcommand{\Xcal}{\mathcal{X}}
\newcommand{\Ycal}{\mathcal{Y}}

\newcommand{\Rr}{\mathscr{R}}

\newcommand{\pen}{{P_e^{(n)}}}

\newcommand{\aep}{{\mathcal{T}_{\epsilon}^{(n)}}}

\newcommand{\Mh}{{\hat{M}}}

\newcommand{\mh}{{\hat{m}}}

\newcommand{\Yt}{{\tilde{Y}}}

\newcommand{\yt}{{\tilde{y}}}

\def\a{\alpha}
\def\b{\beta}

\def\e{\epsilon}

\let\P\relax
\DeclareMathOperator\P{\textsf{P}}

\newcommand{\Bern}{\mathrm{Bern}}

\def\textiid{i.i.d.\@\xspace}
\newcommand\iid{\ifmmode\text{ i.i.d. } \else \textiid \fi}

\def\mathllap{\mathpalette\mathllapinternal}
\def\mathllapinternal#1#2{%
  \llap{$\mathsurround=0pt#1{#2}$}}

\def\clap#1{\hbox to 0pt{\hss#1\hss}}
\def\mathclap{\mathpalette\mathclapinternal}
\def\mathclapinternal#1#2{%
  \clap{$\mathsurround=0pt#1{#2}$}}

\let\oldstackrel\stackrel
\renewcommand{\stackrel}[2]{\oldstackrel{\mathclap{#1}}{#2}}

\renewcommand{\hbar}{h\mathllap{\overline{\vphantom{h}\hphantom{\rule{4.6pt}{0pt}}}\mspace{0.77mu}}}

\catcode`~=11 % make LaTeX treat tilde (~) like a normal character
\newcommand{\urltilde}{\kern -.06em\lower -.06em\hbox{~}\kern .02em}
\catcode`~=13 % revert back to treating tilde (~) as an active character

\newtheorem{theorem}{\textbf{Theorem}}
\newtheorem{lemma}{\textbf{Lemma}}
\newtheorem{corollary}{\textbf{Corollary}}
\newtheorem{proposition}{\textbf{Proposition}}
\newtheorem{definition}{\textbf{Definition}}
\newtheorem{example}{\textbf{Example}}
\newtheorem{remark}{\textbf{Remark}}

\begin{document}

\title{A Note on Broadcast Channels with Stale State Information at the Transmitter}

\author{Hyeji Kim\IEEEauthorrefmark{1}, Yeow-Khiang Chia\IEEEauthorrefmark{2} and Abbas El Gamal\IEEEauthorrefmark{1}

\thanks{\IEEEauthorrefmark{1} Hyeji Kim and Abbas El Gamal are with the Department of Electrical Engineering, Stanford University (email: hyejikim@stanford.edu and abbas@ee.stanford.edu).}\thanks{\IEEEauthorrefmark{2} Yeow-Khiang Chia is with the Institute for Infocomm Research, Singapore (email: yeowkhiang@gmail.com).}
\thanks{ This work was partially supported by Air Force grant FA9550-10-1-0124.   }%
}

\maketitle
\begin{abstract}
This paper shows that the Maddah-Ali--Tse scheme which establishes the symmetric capacity of two example broadcast channels with strictly causal state information at the transmitter is a simple special case of the Shayevitz--Wigger scheme for the broadcast channel with generalized feedback, which involves block Markov coding, compression, superposition coding, Marton coding, and coded time sharing. Focusing on the class of symmetric broadcast channels with state, we derive an expression for the maximum achievable symmetric rate using the Shayevitz--Wigger scheme. We show that the Maddah-Ali--Tse results can be recovered by evaluating this expression for the special case in which superposition coding and Marton coding are not used. We then introduce a new broadcast channel example that shares many features of the Maddah-Ali--Tse examples. We show that another special case of our maximum symmetric rate expression in which superposition coding is also used attains a higher symmetric rate than the MAT scheme. The symmetric capacity of this example is not known, however. 
\end{abstract}

\section{Introduction}

It is well known that a broadcast channel with random state $p(y_1,y_2|x,s)p(s)$ when the state $S$ is known at the decoders can be viewed as a broadcast channel with the same input $X$ but with outputs $(Y_1,S)$ and $(Y_2,S)$~\cite[Chapter 7]{El-Gamal--Kim2010}. If the state is also known {\em strictly causally} at the encoder, i.e., the encoder at time $i$ knows $S^{i-1}$, then the setup can be viewed as a broadcast channel with outputs $(Y_1,S)$ and $(Y_2,S)$ and {\em causal feedback} of part of the outputs. Hence the broadcast channel with state known at the decoders and strictly casually at the encoder is intimately related to the broadcast channel with {\em generalized} feedback~\cite{Shayevitz--Wigger2013}, and it is expected that results for one of these two settings can be readily translated into results for the other.

Dueck was the first to show via an insightful example~\cite{Dueck1980} that feedback can enlarge the capacity region of the broadcast channel. The key idea in Dueck's example is for the encoder to broadcast past common information about the channel obtained through feedback. Even though the channel is memoryless, knowledge of this {\em stale} common information at the decoders helps recover previous messages at a higher rate than without feedback. This key idea has inspired the block Markov coding scheme for the broadcast channel with generalized feedback by Shayevitz and Wigger~\cite{Shayevitz--Wigger2013}. In their scheme, new messages are sent in each transmission block together with {\em refinement} information about the previous messages based on the channel information obtained through feedback. The refinement information is obtained by compressing the previous codewords in a manner similar to the Gray-Wyner system with side information~\cite{Gray--Wyner1974}.  The encoder uses Marton coding, superposition coding, and coded time sharing to encode the messages and the refinement information. Decoding is performed backwards with the refinement information decoded in a block used to decode the messages and the refinement information sent in the previous block.

In a separate line of investigation motivated by fading broadcast channels and network coding, Maddah-Ali and Tse~\cite{Ali--Tse2010} demonstrated via two beautiful examples that strictly causal (stale) state information at the encoder can enlarge the capacity region of the broadcast channel with state when the state is also known at the decoders. In their scheme, which establishes the symmetric capacity for these two examples, transmission is performed over three blocks. In the first block, the message intended for the first receiver is sent at a rate higher than what it can reliably decode. In the second block, the message for the second receiver is sent again at a rate higher than what it can decode. In the third block, refinement information about the messages the depends on the state information from the first two blocks is sent to both receivers to enable them to decode their respective messages. 

In this paper, we show that the Maddah-Ali--Tse (MAT) scheme is a simple special case of a straightforward adaptation of the Shayevitz--Wigger scheme. 
%%%CHANGE FROM HERE
% We focus on a class of symmetric broadcast channels with state that includes the Maddah-Ali--Tse examples as special cases and derive an expression for the maximum symmetric rate achieved using the Shayevitz--Wigger scheme. We then consider the special case 
%{\color{blue}
We consider a class of symmetric broadcast channels with state and derive an expression for the maximum symmetric rate achieved using the Shayevitz--Wigger scheme. We then specialize our result to the subclass of symmetric \emph{deterministic} broadcast channels with state that includes the Maddah-Ali--Tse examples as special cases. We consider the special case of Shayevitz--Wigger scheme
%%% TO HERE
in which superposition coding and Marton coding are not used, henceforth referred to as the {\em time-sharing scheme}, and specialize our expression of the maximum symmetric rate to this case. We show that the maximum symmetric rate for this time-sharing scheme is optimal for the Maddah-Ali--Tse examples and is in fact a simple extension of their scheme. Observing that in both of the Maddah-Ali--Tse examples the channel is deterministic for each state (in addition to being symmetric), we investigate the question of whether the time-sharing scheme is optimal for all such deterministic channels. We construct a new example in which the channel switches between a Blackwell broadcast channel~\cite{blackwell} and a skew symmetric version of it, and show that another special scheme that includes superposition coding, henceforth referred to as the {\em superposition coding scheme}, achieves a higher symmetric rate than the time-sharing scheme. We do not know, however, if the Shayevitz--Wigger scheme in its full generality is optimal for this channel, or for the aforementioned deterministic class in general.

The rest of the paper is organized as follows. In the following section, we provide the needed definitions. In Section~\ref{sect:3}, we adapt the Shayevitz--Wigger scheme to the broadcast channel with stale state information and derive an expression for the maximum achievable symmetric rate when the channel is symmetric. 
%%%CHANGE FROM HERE
In Section~\ref{sect:4}, we specialize this expression to the time-sharing scheme %{\color{blue}
for the symmetric deterministic channels and evaluate the expression to show that the time-sharing scheme is optimal for the Maddah-Ali--Tse examples. In Section~\ref{sect:5}, we specialize our maximum symmetric rate expression to the superposition coding scheme %{\color{blue}
for symmetric deterministic channels and introduce the Blackwell broadcast channel with state example. We show that the maximum symmetric rate using the superposition coding scheme is strictly higher than using the time-sharing scheme. We also obtain an upper bound on the symmetric capacity for this example. 
%%%TO HERE

%---------------------------------------------------------------------------------------------------------------------------------------------------------

\section{Definitions} \label{sect:2}

A 2-receiver DM-BC with generalized feedback consists of an input alphabet $\Xcal$, two output alphabets $(\Ycal_1,\Ycal_2)$, a feedback alphabet $\mathcal{\tilde{Y}}$, and a conditional pmf $p(y_1,y_2,\tilde{y}|x).$ A $(2^{nR_1}, 2^{nR_2},n)$ code for the DM-BC with generalized feedback consists of (i) two message sets $[1:2^{nR_1}]$ and $[1:2^{nR_2}]$; (ii) an encoder that assigns a symbol $x_i(m_1, m_2, \yt^{i-1})$ to each message tuple $(m_1, m_2) \in [1:2^{nR_1}]\times [1:2^{nR_2}]$ and received sequence $\yt^{i-1}$ for $i \in [1:n]$, and (iii) two decoders. Decoder 1 assigns an estimate $\mh_1 \in [1:2^{nR_1}]$ or an error message $e$ to each received sequence $y_1^n$. Decoder 2 assigns $\mh_2 \in [1:2^{nR_2}]$ or an error message $e$ to each received sequence $y_2^n$.

A 2-receiver DM-BC with random state consists of an input alphabet $\Xcal$, two output alphabets $(\Ycal_1,\Ycal_2)$, a discrete memoryless state $S \sim p(s)$, and a conditional pmf $p(y_1, y_2|x,s)$. We consider the case in which the decoders know the state and the encoder knows the state strictly causally (or {\em stale} state in short).  A $(2^{nR_1}, 2^{nR_2},n)$ code for this setup consists of (i) two message sets $[1:2^{nR_1}]$ and $[1:2^{nR_2}]$, (ii) an encoder that assigns a symbol $x_i(m_1, m_2, s^{i-1})$ to each message tuple $(m_1, m_2) \in [1:2^{nR_1}]\times [1:2^{nR_2}]$ and received sequence $s^{i-1}$ for $i \in [1:n]$, and (iii) two decoders. Decoder 1 assigns an estimate $\mh_1 \in [1:2^{nR_1}]$ or an error message $e$ to each received sequence $(y_1^n, s^n)$. Decoder 2 assigns $\mh_2 \in [1:2^{nR_2}]$ or an error message $e$ to each received sequence $(y_2^n, s^n)$.

For both setups, the probability of error is defined as
\begin{align*}
\pen = \P\{\Mh_1 \neq M_1 \text{ or } \Mh_2 \neq M_2\}.
\end{align*}
Similarly, in both cases, a rate tuple $(R_1, R_2)$ is said to be achievable if there exists a sequence of $(2^{nR_1}, 2^{nR_2},n)$ codes such that $P_e^{(n)} \to 0$ as $n \to \infty$. The capacity region is defined as the set of all achievable rate tuples.

\begin{remark}\label{rmk:1} \textnormal{From the above definitions, the latter setup can be viewed as a special case of the former. To see this, let $(X,Y_1,Y_2, \Yt)$ be the random variables associated with the first setup, and $(X',Y_1',Y_2', S)$ be the random variables associated with the second setup. Then set $X = X'$, $Y_1 = (Y_1', S)$, $Y_2 = (Y_2', S)$, and $\Yt = S$. Under this mapping, any coding scheme for the latter case is also a coding scheme for the former case.}
\end{remark}

%%%CHANGE FROM HERE
%This paper will focus on the following special class of channels.
This paper will focus on the following special classes of channels.

\begin{definition}[Symmetric 2-receiver DM-BC with random state]\label{def:1} \textnormal{A 2-receiver DM-BC with random state is said to be symmetric if $\Ycal_1=\Ycal_2= \Ycal$, $\mathcal{S}=\{1, \dots, |\mathcal{S}|\}$,
and there exists a bijective function $\pi: \mathcal{S} \to \mathcal{S}$ such that
\begin{align*}
p_S(s)&=p_S(\pi(s)),\\
p_{Y_1|X,S}(y|x,s)&=p_{Y_2|X,S}(y|x,\pi(s)).
\end{align*}}
\end{definition}

%Added definition
%{\color{blue}
\begin{definition}[Symmetric deterministic 2-receiver DM-BC with random state]\label{def:deterministic}\textnormal{A symmetric 2-receiver DM-BC with random state is said to be deterministic if the outputs are deterministic functions of the input and the state, i.e., $Y_1=y_1(X,S)$ and $Y_2=y_2(X,S)$.}\end{definition}
%}
%The examples in~\cite{Ali--Tse2010} and our new example in Section~\ref{sect:5} all belong to this class of symmetric DM-BC with random state.
The examples in~\cite{Ali--Tse2010} and our new example in Section~\ref{sect:5} all belong to this class of symmetric deterministic DM-BC with random state.
%%%TO HERE

%---------------------------------------------------------------------------------------------------------------------------------------------------------
\section{Maximum symmetric rate for Shayevitz--Wigger scheme}\label{Region} \label{sect:3}

Consider the Shayevitz--Wigger~\cite{Shayevitz--Wigger2013} achievable rate region for the 2-receiver DM-BC with generalized feedback.
\begin{theorem}\label{thm:1}
A rate pair  $(R_1,R_2)$ is achievable for the 2-receiver DM-BC with generalized feedback if it satisfies the following inequalities
%For all choice of auxiliary random variables  $(U_0,U_1,U_2,V_0,V_1,V_2, Q)$, and function $x$ such that
\begin{align}
R_1 \leq& I(U_0,U_1;Y_1,V_1|Q)-I(U_0,U_1,U_2,\tilde{Y};V_0,V_1|Q,Y_1),\nonumber\\
R_2 \leq& I(U_0,U_2;Y_2,V_2|Q)-I(U_0,U_1,U_2,\tilde{Y};V_0,V_2|Q,Y_2),\nonumber\\
R_1+R_2 \leq& I(U_1;Y_1,V_1|Q,U_0)+I(U_2;Y_2,V_2|Q,U_0) + \min_{i \in \{1,2\}}I(U_0;Y_i,V_i|Q)-I(U_1;U_2|Q,U_0)\nonumber\\
& -I(U_0,U_1,U_2,\tilde{Y};V_1|Q,V_0,Y_1)-I(U_0,U_1,U_2,\tilde{Y};V_2|Q,V_0,Y_2)-\max_{i \in \{1,2\}}I(U_0,U_1,U_2,\tilde{Y};V_0|Q,Y_i),\nonumber\\
R_1+R_2 \leq& I(U_0,U_1;Y_1,V_1|Q)+I(U_0,U_2;Y_2,V_2|Q)-I(U_1;U_2|Q,U_0)\nonumber\\
& -I(U_0,U_1,U_2,\tilde{Y};V_0,V_1|Q,Y_1)-I(U_0,U_1,U_2,\tilde{Y};V_0,V_2|Q,Y_2)\nonumber
\end{align}
for some function $x(u_0, u_1, u_2,q)$ and pmf
\[
p(q)p(u_0,u_1, u_2| q)\mathbf{1}_{x=x(u_0,u_1,u_2,q)}p(\yt|x,y_1,y_2,q)p(v_0,v_1,v_2|u_0,u_1, u_2, \yt,q).
\]
\end{theorem}
The following is a simple corollary of the above theorem.

\begin{corollary} \label{coro:1}
A rate pair $(R_1,R_2)$ is achievable for the 2-receiver DM-BC with random state when the state is known at the decoders and strictly causally known at the encoder if it satisfies the following inequalities
\begin{align}
R_1 \leq& I(U_0,U_1;Y_1,V_1|Q,S)-I(U_0,U_1,U_2;V_0,V_1|Q,Y_1,S), \label{bound1}\\
R_2 \leq& I(U_0,U_2;Y_2,V_2|Q,S)-I(U_0,U_1,U_2;V_0,V_2|Q,Y_2,S), \label{bound2}\\
R_1+R_2 \leq& I(U_1;Y_1,V_1|Q,U_0,S)+I(U_2;Y_2,V_2|Q,U_0,S) + \min_{i \in \{1,2\}}I(U_0;Y_i,V_i|Q,S)\nonumber\\
& -I(U_1;U_2|Q,U_0)-I(U_0,U_1,U_2;V_1|Q,V_0,Y_1,S)-I(U_0,U_1,U_2;V_2|Q,V_0,Y_2,S)\nonumber\\
& -\max_{i \in \{1,2\}}I(U_0,U_1,U_2;V_0|Q,Y_i,S),\label{bound3}\\
R_1+R_2 \leq& I(U_0,U_1;Y_1,V_1|Q,S)+I(U_0,U_2;Y_2,V_2|Q,S)-I(U_1;U_2|Q,U_0)\nonumber\\
& -I(U_0,U_1,U_2;V_0,V_1|Q,Y_1,S)-I(U_0,U_1,U_2;V_0,V_2|Q,Y_2,S)\label{bound4}
\end{align}
for some function $x(u_0, u_1, u_2,q)$ and pmf
$p(q)p(u_0,u_1, u_2| q)\mathbf{1}_{x=x(u_0,u_1,u_2,q)}p(s)p(v_0,v_1,v_2|u_0,u_1, u_2, s,q)$. %%%
\end{corollary}

This corollary follows immediately from Remark~\ref{rmk:1}. To be self contained, we give an outline of the coding scheme. The details follow the proof of Theorem~\ref{thm:1} in~\cite{Shayevitz--Wigger2013}. 

The Shayevitz--Wigger scheme uses a block Markov coding in which $b-1$ independent message pairs  $(M_{1,j}, M_{2,j})$ $ \in [1:2^{nR_1}] \times [1:2^{nR_2}]$, $j \in [1:b-1]$, are sent in $b$ $n$-transmission blocks. For simplicity, we describe the scheme only for $Q=\emptyset$ and do not detail the scheme for block $b$. 
\smallskip

\noindent \emph{Codebook generation.} Fix a pmf $p(u_0,u_1,u_2)p(v_0,v_1,v_2|u_0,u_1,u_2,s)$ and a function $x(u_0,u_1,u_2)$. For each block $j \in [1:b-1]$, we independently generate  a codebook for compression and transmission as follows.

\noindent \emph{Codebook generation for compression.}
    \begin{itemize}
    \item Randomly and independently generate $2^{n\tilde{r}_0}$ sequences $v_{0}^n(l_{0,j-1})$, $l_{0,j-1} \in [1:2^{n\tilde{r}_0}]$, each according to $\prod_{t=1}^{n} p_{V_0}(v_{0t})$. Partition the sequences into $2^{nr_{00}}$ equal size superbins indexed by $k_{00,j-1} \in [1:2^{nr_{00}}]$. Further partition each superbin into two subbins, one indexed by $k_{10,j-1} \in [1:2^{nr_{10}}]$ and the other indexed by $k_{20,j-1} \in [1:2^{nr_{20}}]$. 
    \item For each encoder $i\in \{1,2\}$, randomly and independently generate $2^{n\tilde{r}_i}$ codewords $v_i^n(l_{i,j-1})$, $l_{i,j-1} \in [1:2^{n\tilde{r}_i}]$, each according to  $\prod_{t=1}^{n} p_{V_i}(v_{it})$. Partition each set of sequences into $2^{nr_i}$ equal size bins indexed by $k_{i,j-1} \in [1:2^{nr_i}].$
\end{itemize}

\noindent \emph{Codebook generation for transmission. } For each block $j \in [1:b-1],$ we send the refinement (compression) messages together with new messages, i.e., the tuple $(W_{0,j},W_{1,j},W_{2,j})=(K_{00,j-1},(K_{10,j-1}, K_{1,j-1},M_{1,j}),(K_{20,j-1},$ $K_{2,j-1}, M_{2,j}))$. To do so, we use superposition coding and Marton coding to generate the sequence triple $(u_0^n(w_{0,j},$ $w_{1,j},w_{2,j}),$ $ u_1^n(w_{0,j},w_{1,j},w_{2,j}), u_2^n(w_{0,j},w_{1,j},w_{2,j}))\in \aep$ (see~\cite[Chapter 8]{El-Gamal--Kim2010} for details of Marton codebook generation).  %%%Last block - codebook generation!?
%j=1 and j=b. 
\smallskip

Encoding and decoding are described with the help of Table~\ref{table:gs}.

\begin{table}[ht]
\caption{Coding scheme for Corollary~\ref{coro:1}.}
%\centering
\begin{tabular}{c| cc }
\textrm{Block} & 1 & 2\\
\hline\rule{0pt}{15pt}%
	&$\empty$ &$v_0^n(l_{0,1}), v_1^n(l_{1,1}), v_2^n(l_{2,1})$\\[3pt]
$X$ &$(k_{00,0},k_{10,0},k_{20,0}),k_{1,0},k_{2,0}=(1,1,1),1,1$ &$(k_{00,1},k_{10,1},k_{20,1}),k_{1,1},k_{2,1}$\\[3pt]
 	&$u_0^n(w_{0,1},w_{1,1},w_{2,1}),u_1^n(w_{0,1},w_{1,1},w_{2,1}),u_2^n(w_{0,1},w_{1,1},w_{2,1})$ &$u_0^n(w_{0,2},w_{1,2},w_{2,2}), u_1^n(w_{0,2},w_{1,2},w_{2,2}), u_2^n(w_{0,2},w_{1,2},w_{2,2})$\\[3pt]
	&$x^n(u_0^n,u_1^n,u_2^n)$ &$x^n(u_0^n,u_1^n,u_2^n)$\\\\
$Y_1$ &$\hat{m}_{1,1}$ &$\leftarrow (\hat{l}_{1,1}, \hat{k}_{00,1},\hat{k}_{10,1}, \hat{k}_{1,1}), \hat{m}_{1,2}$ \\\\
$Y_2$ &$\hat{m}_{2,1}$ &$\leftarrow (\hat{l}_{2,1},\hat{k}_{00,1},\hat{k}_{20,1}, \hat{k}_{2,1}), \hat{m}_{2,2}$ \\[5pt]
\hline
\end{tabular}
\label{table:gs}
\newline\newline\\
\begin{tabular}{c| ccc}
\textrm{Block} &$\cdots$ &j & $\cdots$ \\
\hline\rule{0pt}{15pt}
	&$\cdots$ &$v_0^n(l_{0,j-1}), v_1^n(l_{1,j-1}), v_2^n(l_{2,j-1})$ & $\cdots$ \\[3pt]
$X$ &$\cdots$ &$(k_{00,j-1},k_{10,j-1},k_{20,j-1}), k_{1,j-1},k_{2,j-1}$& $\cdots$ \\[3pt]
	&$\cdots$ &$u_0^n(w_{0,j},w_{1,j},w_{2,j}), u_1^n(w_{0,j},w_{1,j},w_{2,j}), u_2^n(w_{0,j},w_{1,j},w_{2,j})$ & $\cdots$ \\[3pt]
	&$\cdots$ &$x^n(u_0^n,u_1^n,u_2^n)$ & $\cdots$ \\\\
$Y_1$ &$\cdots$ &$\leftarrow (\hat{l}_{1,j-1}, \hat{k}_{00,j-1},\hat{k}_{10,j-1}, \hat{k}_{1,j-1}), \hat{m}_{1,j}$& $\cdots$ \\\\
$Y_2$ &$\cdots$ &$\leftarrow (\hat{l}_{2,j-1}, \hat{k}_{00,j-1},\hat{k}_{20,j-1}, \hat{k}_{2,j-1}), \hat{m}_{2,j}$& $\cdots$ \\[5pt]
\hline
\end{tabular}
\end{table}

\noindent \emph{Encoding}. In block $j\in [1:b-1]$, the encoder given $s^{n}_{j-1}$ first generates a refinement message tuple $(k_{00,j-1},k_{10,j-1},$\\
$k_{20,j-1}, k_{1,j-1},k_{2,j-1})$ using joint typicality encoding to find $(l_{0,j-1},l_{1,j-1},l_{2,j-1})$ such that 
%$(v_{0}^n(l_{0,j-1}),v_{1}^n(l_{1,j-1}),$ $v_{2}^n(l_{2,j-1}), u_0^n(w_{0,j-1},w_{1,j-1},w_{2,j-1}),u_1^n(w_{0,j-1},w_{1,j-1},w_{2,j-1}),u_2^n(w_{0,j-1},w_{1,j-1},w_{2,j-1}), q_{j-1}^n,s_{j-1}^n) \in \aep$, and the $k$s are the bin indices for the $l$s.
$ (v_{0}^n(l_{0,j-1}),v_{i}^n(l_{i,j-1}),$ $u_0^n(w_{0,j-1},w_{1,j-1},w_{2,j-1}),u_1^n(w_{0,j-1},w_{1,j-1},w_{2,j-1}),u_2^n(w_{0,j-1},w_{1,j-1},w_{2,j-1}), q_{j-1}^n,s_{j-1}^n) \in \aep$ for $i \in \{1,2\}$ and the $k$'s are the bin indices for the $l$'s. The encoder then finds $(u_0^n(w_{0,j},w_{1,j},w_{2,j}), u_1^n(w_{0,j},w_{1,j},w_{2,j}), $ $u_2^n(w_{0,j},$ $w_{1,j},w_{2,j}))$ and transmits $x(u_{0t}(w_{0,j},w_{1,j},w_{2,j}), u_{1t}(w_{0,j},w_{1,j},w_{2,j}), u_{2t}(w_{0,j},w_{1,j},w_{2,j}))$ at time $t\in [1:n]$, where $k_{00,0}=k_{10,0}=k_{20,0}=k_{1,0}=k_{2,0}=1$ by convention. %%%Removed m_{1,b}=m_{2,b}=1.
 
\noindent \emph{Decoding}. The refinements and messages are decoded backward~\cite[Chapter 16]{El-Gamal--Kim2010} starting with block $b$ as described in~\cite{Shayevitz--Wigger2013}.\smallskip

%%%FROM HERE: Haven't changed yet. For the rest of this paper vs. For the rest of this section?
For the rest of this paper, we consider only the symmetric rate for the 2-receiver symmetric DM-BC with random state defined in Section~\ref{sect:2}.
%%%TO HERE

\begin{definition}[Maximum symmetric rate] \textnormal{Let $\Rr$ be the achievable rate region in Corollary~\ref{coro:1} and $R_\mathrm{sym}$ be the  {\em maximum symmetric rate} achievable with the scheme of Corollary~\ref{coro:1}, that is, the supremum of $R$ such that $(R,R) \in \Rr$. Also, let $R_\mathrm{sum}$ be the {\em maximum sum-rate}, that is, the supremum of $R_1+R_2$ such that $(R_1,R_2) \in \Rr$.}
\end{definition}
Because of the restriction to symmetric channels and their symmetric rates, we will need to deal only with auxiliary random variables and functions that satisfy the following. 

\begin{definition}[Symmetric auxiliary random variables] \textnormal{Assume without loss of generality that $\Ucal_1 = \Ucal_2=\Ucal$ and $\Vcal_1 = \Vcal_2=\Vcal$. A set of auxiliary random variables $(U_0,U_1,U_2,V_0,V_1,V_2,Q)$ and function $X=x(U_0,U_1,U_2,Q)$ is said to be symmetric for a symmetric 2-receiver DM-BC with random state  if $\Qcal=\{1, \dots, |\Qcal|\}$ and there exists a bijective function $\tilde{\pi}: \mathcal{Q} \to \mathcal{Q}$ such that
\begin{align*}
p_Q(q)&=p_Q(\tilde{\pi}(q)),\\
p(u_0,u_1,u_2|q)&=p(u_0,u_2,u_1|\tilde{\pi}(q)),\\
x(u_0,u_1,u_2,q)&=x(u_0,u_2,u_1,\tilde{\pi}(q)),\\
p(v_0,v_1,v_2|u_0,u_1,u_2,q,s)&=p(v_0,v_2,v_1|u_0,u_2,u_1,\tilde{\pi}(q),\pi(s))
\end{align*}}
%they satisfy the following properties:
%\begin{enumerate}
%\item $\Qcal=\Qcal_1 \cup \Qcal_2 \cup \Qcal_3$, where $\Qcal_1,\Qcal_2,\Qcal_3$ are disjoint and $|\Qcal_1|=|\Qcal_2|$.
%\item There exists a bijective function $\tilde{\pi}: \Qcal_1 \to \Qcal_2$ such that
%\begin{align*}
%p_Q(q)&=p_Q(\tilde{\pi}(q)) \text{ for } q \in \Qcal_1,\\
%p(u_0,u_1,u_2,v_0,v_1,v_2|s,q)&=p(u_0,u_2,u_1,v_0,v_2,v_1|s,\tilde{\pi}(q)) \text{ for } q \in \Qcal_1.
%\end{align*}
%\item For all $q \in \Qcal_3$, $p(u_0,u_1,u_2,v_0,v_1,v_2|s,q)=p(u_0,u_2,u_1,v_0,v_2,v_1|s,q)$.
%\item  The input
%\[
%X= \begin{cases}
%x(U_0,U_1,U_2,Q) &\text{if } Q \in \Qcal_1,\\
%x(U_0,U_2,U_1,Q) &\text{if } Q \in \Qcal_2,\\
%x(U_0,U_1,U_2,Q)=x(U_0,U_2,U_1,Q) &\text{if } Q \in \Qcal_3.
%\end{cases}\\
%\]
%\end{enumerate}}
\end{definition}
\noindent where $\pi(s)$ is as defined in Definition~\ref{def:1}.
For the symmetric 2-receiver DM-BC, the maximum symmetric rate achievable using the coding scheme of Corollary~\ref{coro:1} can be greatly simplified. To prove this result, we need the following lemma. 

%\begin{lemma} \label{lem:0}
%The maximum sum-rate $R_\mathrm{sum}$ is achievable with a set of symmetric auxiliary random variables.
%\end{lemma}
%The proof of this lemma is given in Appendix.%~\ref{append:1}.

\begin{lemma} \label{lem:1}
Suppose the sum-rate $R_\mathrm{sum}$ is achievable with a set of symmetric auxiliary random variables. Then, $R_\mathrm{sym}=0.5R_\mathrm{sum}$.
\end{lemma}
\begin{IEEEproof} In general, $R_\mathrm{sum} \geq 2 R_\mathrm{sym}$. So we only need to show that if $R_\mathrm{sum}$ is achievable with symmetric auxiliaries, then the rate pair $(0.5R_\mathrm{sum},0.5R_\mathrm{sum})$ is achievable. Note that with symmetric auxiliaries and function, the individual bounds on $R_1$ and $R_2$ in~\eqref{bound1} and~\eqref{bound2} are the same, that is,
\begin{align*}
R_1 &\leq I(U_0,U_1;Y_1,V_1|Q,S)-I(U_0,U_1,U_2;V_0,V_1|Q,Y_1,S)=R_m,\\
R_2 &\leq I(U_0,U_2;Y_2,V_2|Q,S)-I(U_0,U_1,U_2;V_0,V_2|Q,Y_2,S)=R_m.
\end{align*}
Hence, $R_1+R_2 \leq 2R_m$. Since $R_\mathrm{sum}$ is achievable, it must satisfy
$R_\mathrm{sum} \leq 2 R_m$. Hence, $(0.5R_\mathrm{sum},0.5R_\mathrm{sum})$ is achievable, which implies that $R_\mathrm{sym} \ge 0.5 R_\mathrm{sum}$.
\end{IEEEproof}

We are now ready to establish the following simplified expression for the maximum symmetric rate.
\begin{theorem} \label{thm:2} The maximum achievable symmetric rate for the symmetric 2-receiver DM-BC with stale state using the coding scheme of Corollary~\ref{coro:1} is
\begin{align}
\begin{split}\label{thmeq:2}
R_\mathrm{sym}=&\max\min \{I(U_1;Y_1,V_1|Q_\mathrm{sym},U_0,S)+0.5I(U_0;Y_1,V_1|Q_\mathrm{sym},S)-0.5I(U_1;U_2|Q_\mathrm{sym},U_0)\\
&-I(U_0,U_1,U_2;V_1|Q_\mathrm{sym},V_0,Y_1,S)-0.5I(U_0,U_1,U_2;V_0|Q_\mathrm{sym},Y_1,S),\\ &I(U_0,U_1;Y_1,V_1|Q_\mathrm{sym},S)-0.5I(U_1;U_2|Q_\mathrm{sym},U_0)-I(U_0,U_1,U_2;V_0,V_1|Q_\mathrm{sym},Y_1,S)\},
\end{split}
\end{align}
where the maximization is over symmetric auxiliary random variables and functions satisfying the structure in Corollary~\ref{coro:1}.
\end{theorem}

The proof of this theorem is given in Appendix.

%%%FROM HERE
%{\color{blue}
For the rest of this paper, we focus on symmetric \emph{deterministic} 2-receiver DM-BC with random state as defined in Section~\ref{sect:2}.
%%%TO HERE

%%%FROM HERE
%symmetric 2-receiver --> symmetric  {\color{blue}deterministic } 2-receiver
%%%TO HERE
%---------------------------------------------
\section{Time-sharing scheme} \label{sect:4}
% Explain how the coding scheme works and how it is motivated.
In this section, we show that the coding scheme in~\cite{Ali--Tse2010} is a special case of the scheme in~\cite{Shayevitz--Wigger2013} when adapted to the symmetric deterministic DM-BC with random state without superposition coding or Marton coding. We hence refer to this special case as the {\em time-sharing} scheme.
Specifically, we specialize the auxiliary random variables in Theorem~\ref{thm:2} as follows. Let $Q \in \{1,2,3\}$ and $p_Q(1)=p_Q(2)=p$, and $p_Q(3)=1-2 p$, $0 \le p \le 0.5$. Let $p(q,u_0,u_1,u_2)=p(q)p(u_0)p(u_1)p(u_2)$ and $p_{U_1}(u)=p_{U_2}(u)$. Define
\begin{align}
V_1=V_2=V_0&=
	\begin{cases}\label{ts:cond1}
         Y_2  & \text{if }Q=1, 	\\
	Y_1  & \text{if }Q=2, 	\\
	\emptyset  & \text{if }Q=3,
	\end{cases}\\
%V_1&=\begin{cases}
%	Y_2 & \text{if } Q=1,\\
%	\emptyset & \text{if } Q \in \{2,3\},
%	\end{cases}\\
%V_2&=\begin{cases}
%	Y_1 & \text{if }Q=2,\\
%	\emptyset & \text{if } Q \in \{1,3\},
%	\end{cases}\\
X&=  \begin{cases} \label{ts:condend}
	U_1  & \text{if } Q=1,\\
	U_2  & \text{if } Q=2,\\
	U_0  & \text{if } Q=3.
	\end{cases}
\end{align}
Denote the maximum symmetric rate achievable with the above auxiliary random variables identification by $R_\mathrm{sym-ts}$. We now specialize Theorem~\ref{thm:2} to establish the following simplified expression for this maximum symmetric rate.

\begin{proposition}\label{prop:1}
The maximum symmetric rate for the symmetric deterministic 2-receiver DM-BC with stale state using the time-sharing scheme is
\[
R_\mathrm{sym-ts}=\max_{p(x)} \frac{C_1 I(X;Y_1,Y_2|S)}{2C_1+I(X;Y_2|Y_1,S)},
\]
where $C_1 =\max_{p(x)} I(X;Y_1|S)$.
\end{proposition}
\begin{proof}
Substituting from \eqref{ts:cond1} and~\eqref{ts:condend} into~\eqref{thmeq:2}, we obtain
\begin{align}
\begin{split}\label{tssym1}
R_\mathrm{sym-ts}=\max_{p,U_0,U_1}\min\{&pI(U_1;Y_1|S)+(0.5-p)I(U_0;Y_1|S)+0.5pI(U_1;Y_2|Y_1,S),\\
&pI(U_1;Y_1|S)+(1-2p)I(U_0;Y_1|S)\}.
\end{split}
\end{align}

Now we find $p$ and $(U_0, U_1)$ that achieve~\eqref{tssym1}. Since $p \leq 0.5$, the $I(U_0;Y_1|S)$ terms in~\eqref{tssym1} are nonnegative, and without loss of optimality, we can set $U_0=\arg \max I(U_0;Y_1|S)$. Then,
\begin{align}\label{LRpq}
R_\mathrm{sym-ts}&=\max_{p,p(x)}\min\{L(p),R(p)\},
\end{align}
where 
\begin{align*}
L(p)&=pI(X;Y_1|S)+(0.5-p)C_1+0.5pI(X;Y_2|Y_1,S),\\
R(p)&=pI(X;Y_1|S)+(1-2p)C_1. 
\end{align*}

To find $p$ and $X$ that maximize the minimum of the two terms in~\eqref{LRpq}, we first fix $p(X)$ and find $p^*$ that maximizes the min of the two terms in \eqref{LRpq} in terms of $p(X)$. We then optimize $R_\mathrm{sym-ts}$ in $p(X)$.
With $p(X)$ fixed, both $L(p)$ and $R(p)$ are linear functions of $p$, and $L(0) \leq R(0)$ and $L(0.5) \geq R(0.5)$. Thus, $\min\{L(p),R(p)\}$ attains its maximum value at $p^*$ such that $L(p^*)=R(p^*)$, namely,
\begin{align}
p^*=\frac{C_1}{2C_1+I(X;Y_2|Y_1,S)}.\label{optpq}
\end{align}
Replacing $p^*$ in \eqref{LRpq} by \eqref{optpq} completes the proof.
\end{proof}

\begin{remark}\label{rmk:2} \textnormal{Although the Shayevitz--Wigger coding scheme, which achieves the maximum symmetric rate in~\eqref{tssym1}, uses block Markov coding, coded time sharing, and backward decoding, it is not difficult to see that it can be achieved also using the MAT scheme as illustrated in Table~\ref{table:ts}.}
\end{remark}
\begin{remark}\label{rmk:3} \textnormal{ If $\arg\max_{p(x)}I(X;Y_1|S)=\arg\max_{p(x)}I(X;Y_1,Y_2|S)$, then $R_\mathrm{sym-ts}$ can be simplified further to
\begin{align}\label{Rtssym,special}
R_\mathrm{sym-ts}=\frac{C_1C_{1+2}}{C_1+C_{1+2}}, \text{ where }C_1=\max_{p(x)} I(X;Y_1|S), C_{1+2}=\max_{p(x)} I(X;Y_1,Y_2|S),
\end{align}
and is achievable with  $U_1$ and $U_0$ each distributed according to $\arg\max_{p(x)} I(X;Y_1|S)$ and $p^*=C_1/(C_1+C_{1+2})$.}
\end{remark}
\begin{table}[ht]
\caption{Time-sharing scheme.}
\centering
\begin{tabular}{c| ccc}
\textrm{Sub-block} &1 &2 &3\\
\hline\rule{0pt}{15pt}
	 &$\empty$ &$\empty$ &$v_0^n(l_0){=}(y_{2,1}^{pn},y_{1,pn+1}^{2pn},\emptyset_{2pn+1}^n)$\\[3pt]
$X$  &$\empty$ &$\empty$ &$k_{00}$\\[3pt]
	 &$u_{1,1}^{pn}(m_{1})$ &$u_{2,pn+1}^{2pn}(m_{2})$ &$u_{0,2pn+1}^{n}(k_{00})$\\[3pt]
	 & $x_{1}^{pn}{=}u_{1,1}^{pn}$ &$x_{pn+1}^{2pn}{=}u_{2,pn+1}^{2pn}$ &$x_{2pn+1}^{n}{=}u_{0,2pn+1}^{n}$\\\\
$Y_1$     &&&$ \hat{l}_0, \hat{k}_{00}, \hat{m}_{1}$\\\\
$Y_2$     && &$ \hat{l}_0, \hat{k}_{00}, \hat{m}_{2}$\\[5pt]
\hline
\end{tabular}
\label{table:ts}
\end{table}

We now apply the time-sharing scheme to the two examples in~\cite{Ali--Tse2010}, which satisfy the additional condition in Remark~\ref{rmk:3}.

\begin{example}[Broadcast Erasure Channel~\cite{Georgiadis--Tassiulas2009,Ali--Tse2010}] \textnormal{Consider a DM-BC with random state with $X\in \{0,1\}$, $p(y_1,y_2|x)=p(y_1|x)p(y_2|x)$, where $Y_i=X$ with probability $1-\e$ and $Y_i= \mathrm{e}$ with probability $\e$ for $i=1,2$, and $S=(S_1,S_2)$, where $S_i=0$ if $Y_i=X$ and $S_i=1$ if $Y_i=\mathrm{e}$ for $i=1,2$.}

\textnormal{Now, to evaluate the maximum symmetric rate in~\eqref{Rtssym,special}, note that $C_1 = 1-\e$, $C_{1+2} = 1 - \e^2$. Then,
\begin{align}\label{Rtssym,erasure}
R_\mathrm{sym-ts}=\frac{1-\e^2}{2+\e}.
\end{align}
In~\cite{1683207}, an outer bound on the capacity region for this example was obtained based on the observation that this capacity region cannot be larger than that of the physically degraded broadcast channel with input $X$, outputs $Y_1,$ and $(Y_1,Y_2)$, and with causal feedback. Using the same technique, it was shown in~\cite{Ali--Tse2010} that the bound on the symmetric capacity coincides with~\eqref{Rtssym,erasure}.} %%%reference correct?
\end{example}

\begin{example}[Finite Field Deterministic Channel~\cite{Ali--Tse2010}] \textnormal{Consider the DM-BC
\begin{align}
\begin{bmatrix}
Y_1\\
Y_2\end{bmatrix}= H X,\nonumber
\end{align}
where
\[
H=
    \begin{bmatrix}
    h_{11} & h_{12}\\
    h_{21} & h_{22}
    \end{bmatrix},\;
    X=
    \begin{bmatrix}
    X_1\\
    X_2
    \end{bmatrix},\;\text{and }
    S=H.
\]
Assume that  $H$ is chosen uniformly at random from the set of full-rank matrices over a finite field. Further assume that $|\Ycal_1|=|\Ycal_2|=|\Ycal|$.}

\textnormal{Now, to evaluate the maximum symmetric rate in~\eqref{Rtssym,special}, note that $C_1=\log|\Ycal|, C_{1+2}=2\log|\Ycal|.$ Then,
\begin{align}\label{Rtssym,finite}
R_\mathrm{sym-ts}=\frac{2\log|\Ycal|}{3}.
\end{align}}
\textnormal{Using the same converse technique as for Example 1, it was shown in~\cite{Ali--Tse2010} that~\eqref{Rtssym,finite} is the symmetric capacity of this channel.}
\end{example}

%%%CHANGE FROM HERE
%Note that in the above two examples, the channel is deterministic for each state. Is the time-sharing scheme then optimal for all such channels? More precisely, define a \emph{symmetric deterministic 2-receiver DM-BC with random state} to be a symmetric 2-receiver DM-BC with random state in which the outputs are deterministic functions of the input and the state, i.e., $Y_1=y_1(X,S)$ and $Y_2=y_2(X,S)$.  The example in the following section shows that time-sharing scheme is not in general optimal for this class of channels. %{\color{blue} Should we move this definition to the earlier section?}
%{\color{blue}
Note that in the above two examples, the channel is deterministic for each state. Is the time-sharing scheme then optimal for all such channels? The example in the following section shows that time-sharing scheme is not in general optimal for this class of channels.
%%%TO HERE

%------------------------------------------------------
\section{Superposition coding scheme}\label{sect:5}
In the time-sharing scheme, we separately transmit new message and their refinement. 
%%%CHANGE FROM HERE: Didn't change yet. Should we mention that this scheme is for deterministic channels here? We'll mention it in the proposition.
In this section, we consider another special case of the scheme in Corollary~\ref{coro:1} in which we also use superposition coding.
%%%TO HERE
Specifically, we specialize the auxiliary random variables in Theorem~\ref{thm:2} as follows. Let  $Q \in \{1,2\}$ and $P_Q(1)=P_Q(2)=0.5.$ Let $p(q,u_0,u_1,u_2)=p(q)p(u_0)p(u_1|u_0)p(u_2|u_0)$ and $p_{U_1|U_0}(u|u_0)=p_{U_2|U_0}(u|u_0)$. Define
\begin{align}
V_1=V_2=V_0&= \label{sp:cond1}
	\begin{cases}
	Y_2 &\text{if }Q=1,\\
	Y_1 &\text{if }Q=2,\\
	\end{cases}\\
%V_1&=
%	\begin{cases}
%	Y_2  &\text{if }Q=1,\\
%	\emptyset  &\text{if }Q=2,\\
%	\end{cases}\\
%V_2&=
%	\begin{cases}
%	\emptyset &\text{if }Q=1,\\
%	Y_1 &\text{if }Q=2,\\
%	\end{cases}\\
X&=
	\begin{cases}\label{sp:condend}
	U_1  &\text{if }Q=1,\\
	U_2  &\text{if }Q=2.\\
	\end{cases}
\end{align}

The superposition coding scheme is summarized in Table~\ref{table:sp}. 
\begin{table}[ht]
\caption{Superposition coding scheme.}
%\centering
\vspace{5pt}
\begin{tabular}{c| cccc}
\textrm{Block} &\multicolumn{2}{c}{1} &\multicolumn{2}{c}{2}\\
\hline
\textrm{Sub-block} &1 &2 &1 &2\\
\hline\rule{0pt}{15pt}
  	 & \multicolumn{2}{c}{$\empty$} &\multicolumn{2}{c}{$v_0^n(l_{0,1})=(y_{2,1,1}^{n/2},y_{1,1,n/2+1}^n)$} \\[3pt]
$X$  & \multicolumn{2}{c}{$k_{00,0}=1$} &\multicolumn{2}{c}{$k_{00,1}$} \\[3pt]
	 & \multicolumn{2}{c}{$u_{0,1}^n(k_{00,0}),u_{1,1}^{n/2}(m_{1,1}|k_{00,0}), u_{2,n/2+1}^{n}(m_{2,1}|k_{00,0})$} & \multicolumn{2}{c}{$u_{0,1}^n(k_{00,1}),u_{1,1}^{n/2}(m_{1,2}|k_{00,1}), u_{2,n/2+1}^{n}(m_{2,2}|k_{00,1})$}\\[3pt]
	& $x_1^{n/2}{=}u_{1,1}^{n/2}$ & $x_{n/2+1}^{n}{=}u_{2,n/2+1}^{n}$ & $x_1^{n/2}{=}u_{1,1}^{n/2}$ & $x_{n/2+1}^{n}{=}u_{2,n/2+1}^{n}$\\\\
$Y_1$ & & $\hat{m}_{1,1}$ & &$\leftarrow (\hat{l}_{0,1}, \hat{k}_{00,1}), \hat{m}_{1,2}$\\\\
$Y_2$ & & $\hat{m}_{2,1}$ & & $\leftarrow (\hat{l}_{0,1}, \hat{k}_{00,1}), \hat{m}_{2,2}$\\[5pt]
\hline
\end{tabular}
\newline\newline\\
\begin{tabular}{c| c c c c}
\textrm{Block} &$\cdots$ &\multicolumn{2}{c}{j} &$\cdots$\\
\hline
\textrm{Sub-block} &$\cdots$  &1 &2 &$\cdots$\\
\hline\rule{0pt}{15pt}
 &$\cdots$ &\multicolumn{2}{c}{$v_0^n(l_{0,j-1})=(y_{2,j-1,1}^{n/2},y_{1,j-1,n/2+1}^n)$} &$\cdots$\\[3pt]
$X$  &$\cdots$ &\multicolumn{2}{c}{$k_{00,j-1}$} &$\cdots$\\[3pt]
	 &$\cdots$ &\multicolumn{2}{c}{$u_{0,1}^n(k_{00,j-1}),u_{1,1}^{n/2}(m_{1,j}|k_{00,j-1}), u_{2,n/2+1}^{n}(m_{2,j}|k_{00,j-1})$} &$\cdots$\\[3pt]
	 &$\cdots$ & $x_1^{n/2}{=}u_{1,1}^{n/2}$ & $x_{n/2+1}^{n}{=}u_{2,n/2+1}^{n}$ &$\cdots$\\\\
$Y_1$ &$\cdots$  && $\leftarrow (\hat{l}_{0,j-1}, \hat{k}_{00,j-1}), \hat{m}_{1,j}$ &$\cdots$\\\\
$Y_2$ &$\cdots$ & &$\leftarrow (\hat{l}_{0,j-1}, \hat{k}_{00,j-1}), \hat{m}_{2,j}$ &$\cdots$\\[5pt]
\hline
\end{tabular}
\label{table:sp}
\end{table}

Denote the maximum symmetric rate achievable with the above auxiliary random variables identification by $R_\mathrm{sym-sp}$. We now specialize Theorem~\ref{thm:2} to establish the following simplified expression for this maximum symmetric rate.

\begin{proposition}\label{prop:2}
The maximum achievable symmetric rate for the symmetric deterministic 2-receiver DM-BC with stale state using the superposition coding scheme is 
\[
 R_\mathrm{sym-sp}=\max_{p(u_0,x)}\min\{0.5I(X;Y_1|S)+0.25I(X;Y_2|Y_1,U_0,S),0.5I(X;Y_1|S)+0.5I(U_0;Y_1|S)\}.
 \]
%\end{align}}
\end{proposition}
\begin{IEEEproof} This proposition is obtained by substituting from~\eqref{sp:cond1} and~\eqref{sp:condend} into~\eqref{thmeq:2}.
\end{IEEEproof}

We now introduce a new example of a symmetric deterministic broadcast channel with stale state for which the superposition coding scheme outperforms the time-sharing scheme. 

\begin{example}[Blackwell Channel with State]
\textnormal{Consider the symmetric DM-BC with random state depicted in Figure~\ref{fig:bc-blackwell}, where $p_S(1)=p_S(2)=0.5$.}
\begin{figure}[h]
\begin{center}
\psfrag{0}[r]{$0$}
\psfrag{2}[r]{$1$}
\psfrag{1}[r]{$2$}
\psfrag{z}[l]{$0$}
\psfrag{o}[l]{$1$}
\psfrag{X}[c]{$X$}
\psfrag{Y1}[c]{$Y_1$}
\psfrag{Y2}[c]{$Y_2$}
\psfrag{s1}[c]{$S=1$}
\psfrag{s2}[c]{$S=2$}
\includegraphics[scale=0.5]{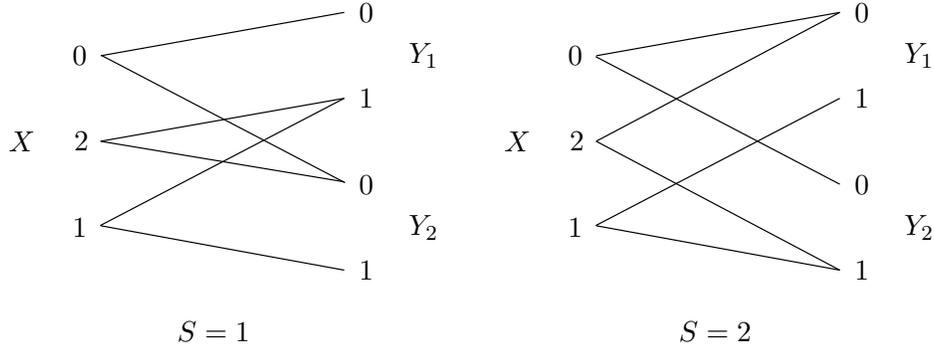}
\caption{Blackwell channel with state.}
\label{fig:bc-blackwell}
\end{center}
\end{figure}

\textnormal{We first evaluate $R_\mathrm{sym-ts}$. Let $U_0\sim \Bern(0.5)$ and $U_1$ and $U_2$ be independently and identically distributed according to $p_{U_1}(0)=p_0,\; p_{U_1}(2)=p_2,\; p_{U_1}(1)=1-p_0-p_2$. We numerically maximize the expression for the maximum symmetric rate in Proposition~\ref{prop:1} in $(p_0,p_2)$ to obtain
\begin{align}
R_\mathrm{sym-ts} &=\max_{p_0,p_2}\frac{H(p_0,1-p_0-p_2,p_2)}{2+0.5(1-p_0)H_b(p_2/(1-p_0))+0.5(1-p_2)H_b(p_0/(1-p_2))}\nonumber\\
				&=0.5989
\end{align}
for $p_0^*=p_2^*=0.37325$. Here, $H(p_0,1-p_0-p_2,p_2)$ is the entropy of $U_1$.}

\textnormal{We now show that superposition coding can do better. Let $U_0 \in \{0,1,2,3\}$ and 
$p_{U_0}(0)=p_{U_0}(1)=q_1, p_{U_0}(2)=p_{U_0}(3)=(1-2q_1)/2, 0 \leq q_1 \leq 0.5$ Let
\begin{align}
%p_{U_1}(u_0)=
%	\begin{cases}
%	q_1 &\textrm{if } u_0=0 \text{ or }1.
%	(1-q_1)/2 &\textrm{if } u_0=2 \text{ or }3.
%	\end{cases}
p_{U_1|U_0}(u|u_0)=p_{U_2|U_0}(u|u_0)=
	\begin{cases}
	\alpha_1(1-\beta_1)  & \textrm{if } (u,u_0)=(0,0) \text{ or } (1,1),\\
	\beta_1 & \textrm{if } (u,u_0)=(2,0) \text{ or }(2,1),\\
	(1-\alpha_1)(1-\beta_1) & \textrm{if } (u,u_0)=(1,0) \text{ or } (0,1),\\
	\alpha_2(1-\beta_2)  & \textrm{if } (u,u_0)=(0,2) \text{ or } (1,3),\\
	\beta_2 & \textrm{if } (u,u_0)=(2,2) \text{ or }(2,3),\\
	(1-\alpha_2)(1-\beta_2) & \textrm{if } (u,u_0)=(1,2) \text{ or } (0,3),
	\end{cases}
\end{align}
and choose $(V_0,V_1,V_2)$ and $X$ as in (\ref{sp:cond1})-(\ref{sp:condend}).}

\textnormal{Maximizing the symmetric rate in Proposition~\ref{prop:2} over $q_1$ and $0 \leq \a_1, \a_2, \b_1, \b_2 \leq 1$, we obtain $R_\mathrm{sym-sp} \geq 0.6103$ at $q_1^*=0.5$, $\alpha_1^*=0.13628$ and $\beta_1^*=0.23025$, which is greater than the symmetric rate achieved using the time-sharing scheme.}

\textnormal{To investigate the optimality of the achievable symmetric rate using superposition coding, we consider the same physically degraded broadcast channel with state in~\cite{Ali--Tse2010} with input $X$ and outputs $(Y_1,Y_2)$ and $Y_2$. The capacity region of this channel is the set of rate pairs $(R_1,R_2)$ such that}
\begin{align}\begin{split}\label{degraded}
R_1 &\leq I(X;Y_1,Y_2|S,U),\\
R_2 &\leq I(U;Y_2|S),
\end{split}\end{align}
\textnormal{where $|\Ucal | \le \min \{ |\Xcal |, |\Ycal_1| |\Scal |, |\Ycal_2| |\Scal | \}+1$.}

\textnormal{Hence, the symmetric capacity is upper bounded as
\begin{align}
C_\mathrm{sym} \le \max_{p(u)p(x|u), |\mathcal{U}| \leq 4} \min\{ I(U;Y_2|S), I(X;Y_1,Y_2|S,U)\}.\label{up1}%\ge 0.653,
\end{align}}

\textnormal{Now we show that this upper bound is strictly less than $2/3$ and is greater than $0.653$. We show that the upper bound is greater than $0.653$ numerically using the substitutions: $U\sim \Bern(0.5), p_{X|U}(0|0)=p_{X|U}(1|1)=0.832,  p_{X|U}(2|0)=p_{X|U}(2|1)=0.168, p_{X|U}(1|0)=p_{X|U}(0|1)=0.$
%$p_U(0)=p_U(1)=0.1, p_U(2)=p_U(3)=0.4$, $p_{X|U}(0|0)=p_{X|U}(1|1)=0.84, p_{X|U}(2|0)=p_{X|U}(2|1)=0.16, p_{X|U}(1|0)=p_{X|U}(0|1)=0, p_{X|U}(0|2)=p_{X|U}(1|3)=0.83, p_{X|U}(2|2)=p_{X|U}(2|3)=0.17, p_{X|U}(1|2)=p_{X|U}(0|3)=0$. 
Thus, the upper bound in~\eqref{up1} is greater than the inner bound using the superposition coding scheme of $R_\mathrm{sym-sp} \geq 0.6103$.}
%\max_{U,X:\, |\mathcal{U}| \leq 4} \min\{ I(U;Y_2|S), I(X;Y_1,Y_2|S,U)\} < 2/3. \label{up2}
%\end{align}

\textnormal{To show that the upper bound in~\eqref{up1} is less than $2/3$, consider
\begin{align*}
\max_{p(u,x)} \min\{ I(U;Y_2|S), I(X;Y_1,Y_2|S,U)\} &\stackrel{(a)}{=} \max_{p(u,x)}\min \{H(Y_2|S)-H(Y_2|S,U), H(Y_1,Y_2|S,U) \} \\
&\stackrel{(b)}{\leq} \max_{p(u,x)}\min \{H(Y_2|S)-H(Y_2|S,U), 2H(Y_2|S,U) \}\\
&\stackrel{(c)}{\leq} \max_{p(x)}\frac{2H(Y_2|S)}{3} \\
&\stackrel{(d)}{\leq} \frac{2}{3},
\end{align*}}\textnormal{where $(a)$ holds because the Blackwell channel with state is deterministic and $(b)$ holds since $H(Y_1,Y_2|S,U) \leq H(Y_1|S,U)+H(Y_2|S,U)=2H(Y_2|S,U)$. Step $(c)$ can be shown as follows. Suppose $H(Y_2|S)-H(Y_2|S,U) > 2H(Y_2|S)/3$,  then $2H(Y_2|S,U) < 2H(Y_2|S)/3$. Therefore at least one of the two terms is less than or equal to $2/3 H(Y_2|S)$. Step $(d)$ holds since $|\Ycal_2|=2$, and equality holds iff $Y_2 \sim \Bern(0.5)$.
Now suppose equality holds for $(b),(c)$, and $(d)$, and then, from equality for $(d)$, $Y_2 \sim \Bern(0.5)$, which implies that $X=Y_1=Y_2 \sim \Bern(0.5)$. Then, from the equality for $(c)$, $H(X|S,U)=1/3$ and from equality for $(b)$, $H(X|S,U)=2H(X|S,U)=0$, which is contradiction. Thus, equality cannot hold for $(b),(c)$, and $(d)$. We conclude that $C_\mathrm{sym} < 2/3$.}

% holds if and only if there exists $p(u,x)$ such that
%\begin{enumerate}
%\item $H(Y_1,Y_2|U,S)=2H(Y_2|U,S), $ \label{cond3} 
%\item $H(Y_2|U,S)=H(Y_2|S)/3$, \label{cond2}
%\item $H(Y_2|S)=1$. \label{cond1}
%\end{enumerate}
%Note that  condition~\ref{cond1} holds if and only if $X \sim \Bern(0.5)$. With $X \sim \Bern(0.5)$, $Y_1=Y_2=X.$ Substituting this to condition~\ref{cond3} and condition~\ref{cond2}, we obtain $H(X|U,S)=2H(X|U,S)=0$ and $H(X|U,S)=H(X|S)/3=1/3$, respectively. This is impossible. Therefore,
%\begin{align*}
%\min \{H(Y_2|S)-H(Y_2|U,S), H(Y_1,Y_2|S,U) \} < 2/3.
%\end{align*}

% %%%Do we want to write that the upper bound is greater than/equal to 0.653?
%We are able to numerically show that the above upper bound is greater than or equal to 0.653. This is achieved using 
%%and  $p_{X|U}(0|0)=0.83, p_{X|U}(2|0)=0.17, p_{X|U}(1|0)=0, p_{X|U}(0|1)=0, p_{X|U}(2|1)=0.17, p_{X|U}(1|1)=0.83.$ 

\end{example}
%-----------------
\section{Conclusion}\label{conc}
We derived a simplified expression for the maximum symmetric rate achievable using the Shayevitz--Wigger scheme for the symmetric broadcast channel with random state when the state is known at the receivers and only strictly causally at the transmitter. We considered a time-sharing special case of the Shayevitz--Wigger scheme and showed that it attains the symmetric capacity of the MAT examples. We then introduced the Blackwell channel with state example and showed that a superposition coding special case of the Shayevitz--Wigger scheme can achieve a higher symmetric rate than the time-sharing scheme. 

There are many open questions that would be interesting to explore further, including the following.
\begin{itemize}
\item We showed that the time-sharing scheme is not optimal for the class of deterministic channels as defined in Section~\ref{sect:2}.  For what general class of channels is it optimal?
\item Is the symmetric rate  achieved using the superposition coding scheme for the Blackwell channel with state example optimal? Can a higher symmetric rate be achieved using Marton coding?
\item For what general class of channels is the symmetric rate achieved using the Shayevitz--Wigger scheme optimal?
\end{itemize}

%\noindent
%1. Is it possible to show that the achievable scheme is symmetric rate optimal for a class of deterministic BC (with state) that includes Tse's examples, but not the blackwell channel?\\
%2. Definition of symmetric DM-BC: can we further generalize the definition of symmetric DM-BC? There are other more definitions of symmetry, but i am not sure if they are useful in this context. For e.g. there is a nice definition for DMBC without state in the paper ``Capacity regions of two new classes of 2-receiver broadcast channels'' by C. Nair.  \\
%3. Can we prove that for any coding scheme that achieves the boundary of the capacity region at $(R_a, R_b)$, if the channel is symmetric, then $(R_b, R_a)$ is also achievable? It seems to be true, we need a proof that comes from operational definitions.\\
%---------------
%\section{Acknowledgments}
%This work was partially supported by Air Force grant FA9550-10-1-0124. 

\appendices

\appendix
\section*{Proof of Theorem~\ref{thm:2}}%\label{append:1}

Suppose $R_\mathrm{sum}$ is achievable with a set of auxiliary random variables $(U_0,U_1,U_2,V_0,V_1,V_2,Q)$ with $Q \in [1:N]$ and a function $X=x(U_0,U_1,U_2,Q)$. Then,
\begin{align}
\begin{split}\label{sum2}
  R_\mathrm{sum}=\min\{&I(U_1;Y_1,V_1|Q,U_0,S)+I(U_2;Y_2,V_2|Q,U_0,S) + \min_{i \in \{1,2\}}I(U_0;Y_i,V_i|Q,S)\\
  &-I(U_1;U_2|Q,U_0)-I(U_0,U_1,U_2;V_1|Q,V_0,Y_1,S)-I(U_0,U_1,U_2;V_2|Q,V_0,Y_2,S)\\
&-\max_{i \in \{1,2\}}I(U_0,U_1,U_2;V_0|Q,Y_i,S), I(U_0,U_1;Y_1,V_1|Q,S)+I(U_0,U_2;Y_2,V_2|Q,S)\\
&-I(U_1;U_2|Q,U_0) -I(U_0,U_1,U_2;V_0,V_1|Q,Y_1,S)-I(U_0,U_1,U_2;V_0,V_2|Q,Y_2,S)\}.
\end{split}
\end{align}
%Hyeji: We shouldn't have two labels for this one eqn. Let's just have one label for the whole thing and refer to first term as "first term inside the minimum in~\eqref{}")

Now we show that this sum-rate is also achievable with a set of symmetric auxiliary random variables.

We first construct the following set of auxiliaries and function:
\begin{align}\begin{split}\label{Qprime}
&Q' \in [N+1:2N],\\%, N \in \{1,2,\dots\}
&p_{Q'}(q) =p_Q(q-N) \text{ for } q \in [N+1:2N],\\
&p_{U_0',U_1',U_2',V_0',V_1',V_2'|S,Q'}(u_0,u_1,u_2,v_0,v_1,v_2|s,q)=p_{U_0,U_1,U_2,V_0,V_1,V_2|S,Q}(u_0,u_2,u_1,v_0,v_2,v_1|\pi(s),q-N),\\
&X=x(U_0',U_2',U_1',Q'-N).\end{split}
\end{align}
Note that the following equalities hold.
\begin{align}
\begin{split}\label{symmetry}
I(U_0',U_1';Y_1,V_1'|Q',S)&=I(U_0,U_2;Y_2,V_2|Q,S),\\
I(U_0',U_1',U_2';V_0',V_1'|Q',Y_1,S)&=I(U_0,U_1,U_2;V_0,V_2|Q,Y_2,S),\\
I(U_1';Y_1,V_1'|Q',U_0,S)&=I(U_2;Y_2,V_2|Q,U_0,S),\\
I(U_0';Y_1',V_1'|Q',S)&=I(U_0;Y_2,V_2|Q,S),\\
I(U_1';U_2'|Q',U_0')&=I(U_1;U_2|Q,U_0),\\
I(U_0',U_1',U_2';V_1'|Q',V_0',Y_1,S)&=I(U_0,U_1,U_2;V_2|Q,V_0,Y_2,S),\\
I(U_0',U_1',U_2';V_0'|Q',Y_1,S)&=I(U_0,U_1,U_2;V_0|Q,Y_2,S).
\end{split}
\end{align}
%{\allowdisplaybreaks
We prove the first equality in~\eqref{symmetry}. 
\begin{align*}
I(U_0',U_1';Y_1,V_1'|Q',S)&=\sum_{s}\sum_{q=N+1}^{2N}p_S(s)p_{Q'}(q)I(U_0',U_1';Y_1,V_1'|Q'=q,S=s)\\
&=\sum_{\pi(s)}\sum_{q=N+1}^{2N}p_S(\pi(s))p_Q(q-N)I(U_0,U_2;Y_2,V_2|Q=q-N,S=\pi(s))\\
&=\sum_{s}\sum_{q=1}^{N}p_S(s)p_Q(q)I(U_0,U_2;Y_2,V_2|Q=q,S=s)\\
&=I(U_0,U_2;Y_2,V_2|Q,S).
\end{align*}
The rest of the equalities in~\eqref{symmetry} can be proved in a similar manner.

Now we compose a new set of auxiliaries that ``time-share" between $Q$ and $Q'$. Let
\begin{align}
Q_\mathrm{sym}= \begin{cases}
	Q  & \textrm{ with probability } 0.5, 	\nonumber\\
	Q'  & \textrm{ with probability } 0.5. 	
	\end{cases}
\end{align}

It follows from~\eqref{symmetry} that,  for $i \in \{1,2\},$
\begin{align}
\begin{split}\label{symmetry2nd}
I(U_0,U_i;Y_i,V_i|Q_\mathrm{sym},S)&=0.5\sum_{i=1,2}I(U_0,U_i;Y_i,V_i|Q,S),\\
I(U_i;Y_i,V_i|Q_\mathrm{sym},U_0,S)&=0.5\sum_{i=1,2}I(U_i;Y_i,V_i|Q,U_0,S),\\
I(U_0;Y_i,V_i|Q_\mathrm{sym},S)&=0.5\sum_{i=1,2}I(U_0;Y_i,V_i|Q,S),\\
I(U_1;U_2|Q_\mathrm{sym},U_0)&=I(U_1;U_2|Q,U_0),\\
I(U_0,U_1,U_2;V_i|Q_\mathrm{sym},V_0,Y_i,S)&=0.5\sum_{i=1,2}I(U_0,U_1,U_2;V_i|Q,V_0,Y_i,S),\\
I(U_0,U_1,U_2;V_0|Q_\mathrm{sym},Y_i,S)&=0.5\sum_{i=1,2}I(U_0,U_1,U_2;V_0|Q,Y_i,S).
\end{split}
\end{align}

The sum-rate achievable with $Q_\mathrm{sym}$ is
\begin{align}
\begin{split}\label{symsum}
  \mathrm{SR}= \min \{&I(U_1;Y_1,V_1|Q_\mathrm{sym},U_0,S)+I(U_2;Y_2,V_2|Q_\mathrm{sym},U_0,S) + \min_{i \in \{1,2\}}I(U_0;Y_i,V_i|Q_\mathrm{sym},S) \\
&\qquad~~~-I(U_1;U_2|Q_\mathrm{sym},U_0)-I(U_0,U_1,U_2;V_1|Q_\mathrm{sym},V_0,Y_1,S)-I(U_0,U_1,U_2;V_2|Q_\mathrm{sym},V_0,Y_2,S)\\
&\qquad~~~-\max_{i \in \{1,2\}}I(U_0,U_1,U_2;V_0|Q_\mathrm{sym},Y_i,S),\\
&I(U_0,U_1;Y_1,V_1|Q_\mathrm{sym},S)+I(U_0,U_2;Y_2,V_2|Q_\mathrm{sym},S)-I(U_1;U_2|Q_\mathrm{sym},U_0)\\
&\qquad~~~-I(U_0,U_1,U_2;V_0,V_1|Q_\mathrm{sym},Y_1,S)-I(U_0,U_1,U_2;V_0,V_2|Q_\mathrm{sym},Y_2,S)\}.
\end{split}
\end{align}

%Now we use~\eqref{symmetry} to rewrite~\eqref{symsum}.

We now show that $\mathrm{SR} \geq R_\mathrm{sum}$. By using~\eqref{symmetry2nd}, it can be easily shown that the second term inside the minimum in~\eqref{symsum} is the same as the second term inside the minimum in~\eqref{sum2}. We now show that the first term inside the minimum in \eqref{symsum} is greater than or equal to the first term inside the minimum in \eqref{sum2}. We start with the first term in \eqref{symsum}.
{\allowdisplaybreaks
\begin{align*}
&I(U_1;Y_1,V_1|Q_\mathrm{sym},U_0,S)+I(U_2;Y_2,V_2|Q_\mathrm{sym},U_0,S) + \min_{i \in \{1,2\}}I(U_0;Y_i,V_i|Q_\mathrm{sym},S)-I(U_1;U_2|Q_\mathrm{sym},U_0)\nonumber\\
&\qquad~-I(U_0,U_1,U_2;V_1|Q_\mathrm{sym},V_0,Y_1,S)-I(U_0,U_1,U_2;V_2|Q_\mathrm{sym},V_0,Y_2,S)-\max_{i \in \{1,2\}}I(U_0,U_1,U_2;V_0|Q_\mathrm{sym},Y_i,S)\\
& \stackrel{(a)}{=}I(U_1;Y_1,V_1|Q,U_0,S)+I(U_2;Y_2,V_2|Q,U_0,S) + 0.5\sum_{i=1,2}I(U_0;Y_i,V_i|Q,S)-I(U_1;U_2|Q,U_0)\\
&\qquad~-I(U_0,U_1,U_2;V_1|Q,V_0,Y_1,S)-I(U_0,U_1,U_2;V_2|Q,V_0,Y_2,S)-0.5\sum_{i=1,2}I(U_0,U_1,U_2;V_0|Q_\mathrm{sym},Y_i,S)\\
&\geq I(U_1;Y_1,V_1|Q,U_0,S)+I(U_2;Y_2,V_2|Q,U_0,S) + \min_{i \in 1,2} I(U_0;Y_i,V_i|Q,S)-I(U_1;U_2|Q,U_0)\\
&\qquad~-I(U_0,U_1,U_2;V_1|Q,V_0,Y_1,S)-I(U_0,U_1,U_2;V_2|Q,V_0,Y_2,S)-\max_{i \in \{1,2\}}I(U_0,U_1,U_2;V_0|Q,Y_i,S),
%& \stackrel{(b)}{=}I(U_1;Y_1,V_1|Q,U_0,S)+I(U_2;Y_2,V_2|Q,U_0,S) + 0.5\min_{i \in \{1,2\}}\{0.5I(U_0;Y_i,V_i|Q,S)+0.5I(U_0';Y_i',V_i'|Q',S)\}\\
% &\qquad~-I(U_1;U_2|Q,U_0)-I(U_0,U_1,U_2;V_1|Q,V_0,Y_1,S)-I(U_0,U_1,U_2;V_2|Q,V_0,Y_2,S)\\
% &\qquad~-\max_{i \in \{1,2\}}\{0.5I(U_0,U_1,U_2;V_0|Q,Y_i,S)+0.5I(U_0',U_1',U_2';V_0|Q',Y_i,S)\}\\
%& \stackrel{(c)}{\geq} I(U_1;Y_1,V_1|Q,U_0,S)+I(U_2;Y_2,V_2|Q,U_0,S) + \min_{i \in \{1,2\}}0.5I(U_0;Y_i,V_i|Q,S)+\min_{i \in \{1,2\}}0.5I(U_0';Y_i,V_i'|Q',S)\\
% &\qquad~-I(U_1;U_2|Q,U_0)-I(U_0,U_1,U_2;V_1|Q,V_0,Y_1,S)-I(U_0,U_1,U_2;V_2|Q,V_0,Y_2,S)\\
% &\qquad~-\max_{i \in \{1,2\}}0.5I(U_0,U_1,U_2;V_0|Q,Y_i,S)-\max_{i \in \{1,2\}}0.5I(U_0',U_1',U_2';V_0'|Q',Y_i,S)\\
%& \stackrel{(d)}{=}I(U_1;Y_1,V_1|Q,U_0,S)+I(U_2;Y_2,V_2|Q,U_0,S) + \min_{i \in \{1,2\}}I(U_0;Y_i,V_i|Q,S)-I(U_1;U_2|Q,U_0)\\
%&\qquad~-I(U_0,U_1,U_2;V_1|Q,V_0,Y_1,S)-I(U_0,U_1,U_2;V_2|Q,V_0,Y_2,S)-\max_{i \in \{1,2\}}I(U_0,U_1,U_2;V_0|Q,Y_i,S),
\end{align*}
where $(a)$ holds from~\eqref{symmetry2nd}. % the concavity of the $\min$ function and the convexity of the $\max$.
Therefore, the maximum sum-rate is achievable with symmetric auxiliary random variables.}

By lemma~\ref{lem:1}, $R_\mathrm{sym}$ can be written as
\begin{align}
\begin{split}\label{fromLemma}
R_\mathrm{sym}=& \frac{1}{2}\max \min \{\sum_{i=1,2}I(U_i;Y_i,V_i|Q_\mathrm{sym},U_0,S)+ \min_{i \in \{1,2\}}I(U_0;Y_i,V_i|Q_\mathrm{sym},S)-I(U_1;U_2|Q_\mathrm{sym},U_0)\\
&-\sum_{i=1,2}I(U_0,U_1,U_2;V_i|Q_\mathrm{sym},V_0,Y_i,S)-\max_{i \in \{1,2\}}I(U_0,U_1,U_2;V_0|Q_\mathrm{sym},Y_i,S),\\
&\sum_{i=1,2} I(U_0,U_i;Y_i,V_i|Q_\mathrm{sym},S)-I(U_1;U_2|Q_\mathrm{sym},U_0)-\sum_{i=1,2}I(U_0,U_1,U_2;V_0,V_i|Q_\mathrm{sym},Y_i,S)\},
\end{split}
\end{align}
where the maximization is over symmetric auxiliaries and functions satisfying the structure in Corollary~\ref{coro:1}.
Using \eqref{symmetry2nd}, $R_\mathrm{sym}$ can be further simplified to 
\begin{align*}
%\begin{split}\label{Rsym}
R_\mathrm{sym}=&\max\min \{I(U_1;Y_1,V_1|Q_\mathrm{sym},U_0,S)+0.5I(U_0;Y_1,V_1|Q_\mathrm{sym},S)-0.5I(U_1;U_2|Q_\mathrm{sym},U_0)\\
&-I(U_0,U_1,U_2;V_1|Q_\mathrm{sym},V_0,Y_1,S)-0.5I(U_0,U_1,U_2;V_0|Q_\mathrm{sym},Y_1,S),\\ &I(U_0,U_1;Y_1,V_1|Q_\mathrm{sym},S)-0.5I(U_1;U_2|Q_\mathrm{sym},U_0)-I(U_0,U_1,U_2;V_0,V_1|Q_\mathrm{sym},Y_1,S)\}.
%\end{split}
\end{align*}
%where the maximization is over symmetric auxiliary random variables and functions satisfying the structure in Corollary~\ref{coro:1}.

\bibliographystyle{IEEEtran}
\bibliography{bcstate}
%\begin{thebibliography}{99}
%\bibitem{Shayevitz} O. Shayevitz and M. Wigger, On the Capacity of the Discrete Memoryless Broadcast Channel with Feedback",
%\bibitem{erasure} L. Georgiadis and L. Tassiulas, •À?Broadcast erasure channel with feedback - capacity and algorithms, in Workshop on Network Coding, Theory, and
%Applications, Lausanne, June 2009, pp. 54 •À? 61.
%\bibitem{stale} Mohammad Ali Maddah-Ali and David Tse, ''Completely stale channel state information is still very useful'', Online at ArXiv: http://arxiv.org/pdf/1010.1499v3.pdf
%\bibitem{NIT} A. El Gamal and Y.-H. Kim, ``Network Information Theory", Cambridge University Press, 2011.
%Arxiv preprint arXiv:1012.6012, Nov. 2011.
%\end{thebibliography}

%\cleardoublepage

\end{document}